\documentclass[letterpaper,11pt]{article}
\usepackage[utf8]{inputenc}

\pdfoutput=1

\usepackage{amsmath, amsthm, amssymb, thm-restate}
\usepackage{algorithmicx}
\usepackage[table,xcdraw]{xcolor}
\usepackage[ruled,vlined,linesnumbered]{algorithm2e}
\usepackage{setspace}
\usepackage{mathtools}
\usepackage[numbers]{natbib}
\usepackage{comment} 
\usepackage{tcolorbox} 
\usepackage{xfrac}
\usepackage{hyperref}
\usepackage{multirow}
\usepackage{caption}
\usepackage{bm}
\usepackage{newfloat}
\usepackage{enumitem}
\usepackage{dblfloatfix} 
\usepackage{wrapfig}

\usepackage{fancyhdr}

\usepackage[margin=1in]{geometry}

\allowdisplaybreaks

\definecolor{mygreen}{RGB}{10,150,110}
\definecolor{myred}{RGB}{150,10,20}

\hypersetup{
     colorlinks=true,
     citecolor= mygreen,
     linkcolor= myred
}

\renewcommand{\epsilon}{\varepsilon}

\DeclareMathOperator{\E}{\ensuremath{\normalfont \textbf{E}}}

\newcommand{\sbcomment}[1]{{\textcolor{blue}{[\textbf{Soheil:} #1]}}}

\newcommand{\mrcomment}[1]{{\textcolor{red}{[\textbf{Mohammad:} #1]}}}

\newcommand{\Aviad}[1]{{\textcolor{purple}{[\textbf{Aviad:} #1]}}}

\newcommand{\hiddencomment}[1]{}

\newcommand{\yes}[0]{\ensuremath{\mathsf{YES}}}
\newcommand{\no}[0]{\ensuremath{\mathsf{NO}}}
\newcommand{\yesdist}[0]{\ensuremath{\mathcal{D}_{\yes}}}
\newcommand{\nodist}[0]{\ensuremath{\mathcal{D}_{\no}}}

\DeclareMathOperator{\poly}{poly}
\DeclareMathOperator{\polylog}{polylog}

\usepackage[noabbrev,nameinlink]{cleveref}
\crefname{lemma}{Lemma}{Lemmas}
\crefname{theorem}{Theorem}{Theorems}
\crefname{property}{Property}{Properties}
\crefname{claim}{Claim}{Claims}
\crefname{result}{Result}{Results}
\crefname{definition}{Definition}{Definitions}
\crefname{observation}{Observation}{Observations}
\crefname{proposition}{Proposition}{Propositions}
\crefname{assumption}{Assumption}{Assumptions}
\crefname{line}{Line}{Lines}
\crefname{figure}{Figure}{Figures}
\creflabelformat{property}{(#1)#2#3}
\crefname{equation}{}{}
\crefname{section}{Section}{Sections}
\crefname{appendix}{Appendix}{Appendices}
\crefname{algCounter}{Algorithm}{Algorithms}
\Crefname{algCounter}{Algorithm}{Algorithms}

\newtheorem{theorem}{Theorem}

\newtheorem{lemma}{Lemma}[section]
\newtheorem{proposition}[lemma]{Proposition}
\newtheorem{corollary}[lemma]{Corollary}

\newtheorem{definition}[lemma]{Definition}
\newtheorem{claim}[lemma]{Claim}

\newtheorem{observation}[lemma]{Observation}
\newtheorem{remark}{Remark}
\newtheorem*{remark*}{Remark}

\definecolor{mylightgray}{RGB}{230,230,230}

\algnewcommand{\IIf}[2]{\textbf{if} #1 \textbf{then} #2}
\algnewcommand{\EndIIf}{\unskip\ \algorithmicend\ \algorithmicif}

\newenvironment{whitetbox}{
\par\addvspace{0.1cm}
\begin{tcolorbox}[width=\textwidth,
                  boxsep=5pt,
                  left=1pt,
                  right=1pt,
                  top=2pt,
                  bottom=2pt,
                  boxrule=1pt,
                  arc=0pt,
                  colframe=black,
                  colback=white
                  ]
}{
\end{tcolorbox}
}

\newcounter{algCounter}

\makeatletter
\renewcommand{\paragraph}{%
  \@startsection{paragraph}{4}%
  {\z@}{10pt}{-1em}%
  {\normalfont\normalsize\bfseries}%
}
\makeatother

\makeatletter
\patchcmd{\@algocf@start}
  {-1.5em}
  {0pt}
  {}{}
\makeatother

 \title{Local Computation Algorithms for Maximum Matching:\\ New Lower Bounds}

 \author{
 Soheil Behnezhad\\{\em Northeastern University} \and
 Mohammad Roghani\\{\em Stanford University} \and
 Aviad Rubinstein\\{\em Stanford University}}

\date{}

\begin{document}

\maketitle

\thispagestyle{empty}

\begin{abstract}
    We study {\em local computation algorithms} (LCA) for maximum matching. An LCA does not return its output entirely, but reveals parts of it upon query. For matchings, each query is a vertex $v$; the LCA should return whether $v$ is matched---and if so to which neighbor---while spending a small time per query.

    \smallskip\smallskip
     In this paper, we prove that any LCA that computes a matching that is at most an additive of $\epsilon n$ smaller than the maximum matching in $n$-vertex graphs of maximum degree $\Delta$ must take at least $\Delta^{\Omega(1/\varepsilon)}$ time. This comes close to the existing upper bounds that take $(\Delta/\epsilon)^{O(1/\epsilon^2)}\polylog (n)$ time.

     \smallskip\smallskip
     In terms of sublinear time algorithms, our techniques imply that any algorithm that estimates the size of maximum matching up to an additive error of $\epsilon n$ must take $\Delta^{\Omega(1/\epsilon)}$ time.  This negatively resolves a decade old open problem of the area (see Open Problem 39 of sublinear.info) on whether such estimates can be achieved in $\poly(\Delta/\epsilon)$ time.
\end{abstract}

{
\clearpage
\hypersetup{hidelinks}
\vspace{1cm}
\renewcommand{\baselinestretch}{0.1}
\setcounter{tocdepth}{2}
\thispagestyle{empty}
\clearpage
}

\setcounter{page}{1}
\section{Introduction}
Over the last two decades, there has been growing interest in the development and study of algorithms that do not return their output in whole, but instead return parts of it upon query. These algorithms, known as {\em local computation algorithms} (LCAs), have emerged as a powerful tool in the field of sublinear algorithms, enabling efficient processing of massive data for problems where the output is too large to be stored or reported in its entirety. In this work, we study LCAs for the {\em maximum matching} problem. This is one of the most intensively studied problems in the literature of LCAs. We first overview the model and prior work, then describe our contribution.

\paragraph{The LCA Model:} Local computation algorithms were formalized in the works of \citet{RubinfeldTVX11} and \citet{AlonRVX12}. For graph problems, an LCA can access the graph through adjacency list queries. That is, by specifying a vertex $v$ and an integer $i$, the LCA is either given the $i$-th neighbor of vertex $v$ or ``$\perp$'' if $v$ has less than $i$ neighbors. When questioned about a vertex $v$, an LCA can make queries to the graph and a tape of randomness to compute its output on that vertex. In the case of matchings---the focus of our paper---this output is whether the questioned vertex is matched and if so to which of its neighbors. These answers must be independent of the order of questions, meaning that the LCA should be able to produce consistent answers even if multiple vertices are questioned in parallel.  The worst-case number of queries an LCA conducts to answer any single question is the measure of its complexity.

As standard in the literature, we say a matching $M$ in graph $G$ provides an $(\alpha, \epsilon n)$-approximation if $|M| \geq \alpha \mu(G) - \epsilon n$, where $\mu(G)$ is the size of the maximum matching in $G$.

\paragraph{Known Algorithms:}  The (approximate) maximum matching problem has been studied extensively both in the literature of LCAs \cite{RubinfeldTVX11,AlonRVX12,levironitt,KapralovSODA20,Ghaffari-FOCS22,Bhattacharya-ArXiv23} and the closely related model of sublinear time algorithms \cite{ParnasRon07,NguyenOnakFOCS08,YoshidaYISTOC09,behnezhad2021,ChenICALP20,BehnezhadRRS-SODA23,behnezhad23,BKS22}. Earlier works in both models only focused on bounded degree graphs where the maximum degree $\Delta$ is constant.  
There has been a sequence of improvements on LCAs \cite{ReingoldV16,levironitt,GhaffariUittoSODA19,Ghaffari-FOCS22,levironitt,KapralovSODA20,Bhattacharya-ArXiv23}.
The best-known algorithm for a $(1, \epsilon n)$-approximate maximum matching is due to \citet*{levironitt} which adapts the elegant sublinear time algorithm of \citet*{YoshidaYISTOC09} to the LCA model, achieving a running time of $(\Delta/\epsilon)^{O(1/\epsilon^2)}\poly\log(n)$ per query. Note that the LCA of \cite{levironitt,YoshidaYISTOC09} runs in time $\poly(\Delta, \log n)$ whenever $\epsilon$ is constant. Thus it runs efficiently even in the case of ``graphs of non-constant degree'' \cite{levironitt}. If instead of a $(1, \epsilon n)$-approximation we desire a $(1/2, 0)$ approximation, then this can be done in $O(\Delta \poly\log n)$ time \cite{behnezhad2021}.\footnote{We note that the LCA model is not directly studied in \cite{behnezhad2021}, but the abovementioned bound follows as a corollary of \cite{behnezhad2021}.}


We also note that in the orthogonal dense regime (with adjacency matrix queries), \citet*{Bhattacharya-ArXiv23} showed in a very recent paper that there is an LCA with complexity $n^{2-\Omega_\epsilon(1)}$ that computes a $(1, \epsilon n)$-approximate matching.

\paragraph{Known Lower Bounds:} Only two results in the literature give lower bounds for LCAs approximating maximum matching. The first one, due to \citet*{ParnasRon07} from 2007, proves that any LCA computing a constant approximation of maximum matching needs to spend $\Omega(\Delta)$ time. The essence of the lower bound of \cite{ParnasRon07} is a construction, where each vertex has degree $\Theta(\Delta)$ and has only one ``important'' edge that has to be in any constant approximate matching. Thus, any LCA that reports a constant approximate matching must scan a constant fraction of neighbors of the vertex being queried to find this important edge, implying the claimed $\Omega(\Delta)$ lower bound. Note that this approach cannot possibly result in an $\omega(\Delta)$ lower bound.

The second, more recent, result by \citet*{behnezhad23} breaks this linear in $\Delta$ barrier. They gave a construction with maximum degree $\Delta = \Theta(n)$, on which any $(2/3+\Omega(1), \epsilon n)$-approximate algorithm must spend at least $n^{1.2-o(1)}$ time. While the result of \cite{behnezhad23} is stated for sublinear time algorithms, it carries over to the LCA model as well implying a lower bound of $\Delta^{1.2-o(1)}$ for any LCA obtaining a $(2/3+\Omega(1), \epsilon n)$-approximation of maximum matching. The key to the lower bound of \cite{behnezhad23} is a {\em correlation decay} based argument that shows queries far away from a vertex $v$ do not help finding the ``important edge'' of $v$, and one has to explore $\Delta^{1.2-o(1)}$ neighbors in the 2-hop of $v$ to find this edge. We note that the construction of \cite{behnezhad23} can be easily solved if one collects the whole 2-hop neighborhood of the queried vertex, thus it does not lead to $\omega(\Delta^2)$ lower bounds even for much larger than 2/3 approximations.

\paragraph{Our Contribution:} In this paper, we prove a new lower bound on the complexity of LCAs for $(1, \epsilon n)$-approximate matchings. We show that:

\begin{theorem}\label{thm:main}
Let $\epsilon \leq 0.01$. For any choice of $\log^4 n \leq \Delta \leq n^\epsilon$, there is an $n$-vertex bipartite graph $G$ of maximum degree $\Delta$ such that any LCA that with probability at least $0.51$ computes a $(1, \epsilon n)$-approximate maximum matching of $G$ must make at least  $\Delta^{\Omega(1/\varepsilon)}$ queries to $G$.
\end{theorem}

\begin{remark}
The lower bound of \cref{thm:main} holds even if the queried vertices are chosen uniformly at random from a set of $\Theta(n)$ vertices, and even if the answers produced by the LCA can be a function of the order of queries (i.e. if the LCA is not query oblivious).
\end{remark}

\cref{thm:main} significantly improves prior lower bounds and shows that a large polynomial dependence on $\Delta$ is necessary for any LCA computing a $(1, \epsilon n)$-approximate maximum matching. It also comes close to the existing $\Delta^{O(1/\epsilon^2)}\polylog(n)$ time LCAs of \cite{levironitt,YoshidaYISTOC09}, showing that these algorithms are not far from optimal.

\paragraph{Implications for $(1, o(1))$-Approximations:} Obtaining a $\poly(\Delta, \log n)$ bound is a natural target for LCAs (see \cite{Ghaffari-FOCS22} and its references). Another implication of \cref{thm:main}, by setting $\epsilon = o(1)$, is that no such $\poly(\Delta, \log n)$ time LCA exists for computing a $(1, o(n))$-approximate maximum matching.

\paragraph{Implications for Sublinear-Time Algorithms:} Our construction also has implications in the sublinear time model, where the algorithm is provided adjacency list access to the graph and is only required to return an estimate of the size of maximum matching.

In this sublinear time model, \citet*{YoshidaYISTOC09} showed there exists an  $\Delta^{O(1/\epsilon^2)}/\epsilon^2$ time algorithm providing a $(1, \epsilon n)$-approximation of maximum matching size. Whether there exists a $\poly(\Delta/\epsilon)$ time algorithm has remained open for more than a decade. See, in particular, Problem 39 on sublinear.info.\footnote{\url{https://sublinear.info/index.php?title=Open_Problems:39}} Our next \cref{thm:sublinear} negatively resolves this question by showing that $\Delta^{\Omega(1/\epsilon)}$ time is necessary.

\begin{theorem}\label{thm:sublinear}
    Let $\epsilon \leq 0.01$. For any choice of $\log^4 n \leq \Delta \leq n^\epsilon$, there is an $n$-vertex bipartite graph $G$ of maximum degree $\Delta$ such that any randomized algorithm that with probability at least $0.51$ provides a $(1, \epsilon n)$-approximation for the size of the maximum matching in $G$ must make $\Delta^{\Omega(1/\varepsilon)}$ adjacency list queries.
\end{theorem}

Compared to prior lower bounds, several substantially new ideas are needed in the proof of \cref{thm:main}. The main novelty of our proof is a new notion of {\em delusive vertices}. These are a total of $\Theta(\epsilon n)$ vertices in the graph decomposed into $O(1/\epsilon)$ levels of $O(\epsilon^2 n)$ vertices each that essentially do not participate in a $(1, \epsilon n)$-approximate maximum matching, but distinguishing them from those vertices that do participate in the matching turns out to require a large number of queries. We present a detailed overview of these delusive vertices and our techniques in \cref{sec:techniques}.


\paragraph{Paper Organization:}
We present a high-level overview of our lower bound in \cref{sec:techniques}. \cref{sec:prelim} overviews preliminaries, notation, and some basic tools from the literature. \cref{sec:input-distribution} formalizes the construction of our hard instance. \cref{sec:no-cycle} reduces the problem to a certain label guessing game on trees. Finally, in \cref{sec:indis-label-root,sec:indis-yes-no} we prove the lower bound on the complexity of this label guessing game, wrapping up our proof of \cref{thm:main}.



\section{A High-Level Overview of Our Lower Bound}\label{sec:techniques}

In this section, we present a high-level and informal overview of our lower bound of \cref{thm:main}, deferring the formal proofs to the forthcoming sections.

\subsection{The Input Graph}

We start by describing the input distribution. As the final construction might seem strange at the first glance, we present it step by step, gradually adding all the ingredients that are needed for the final proof. Note that the degree of construction outlined in the technical overview differs slightly from the actual construction, but this overview contains all the essential ideas.

\paragraph{Step 1 --- The Core:} The first step is simple and intuitive. The ``core'' of our input graph consists of a set $S$ of vertices of degree 1. The core, in addition, has $k = \Theta(1/\epsilon)$ vertex subsets $A_i, B_i$ for $i \in [k]$. There are two types of edges in the core as illustrated in \cref{fig:img1} for $k=3$. There are `dense blocks' of $d$-regular graphs between $A_i$ and $B_i$ for any $i \in [k]$. Additionally, there are `special edges' perfectly matching $S$ to $B_1$, $A_i$ to $B_{i+1}$ for any $i \in [k-1]$, and $A_k$ to $A_k$.

\begin{figure*}
    \centering
    \includegraphics[width=1\textwidth]{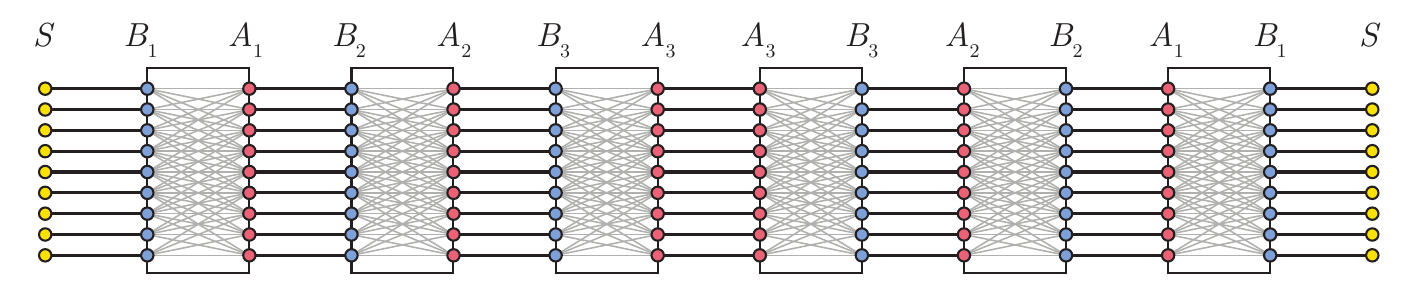}
    \caption{Core of the construction when $k = 3$.}
    \label{fig:img1}
\end{figure*}

Note that the special edges combined form a maximum matching of the core. Importantly, any $(\frac{2k + 1/2}{2k+1} \sim 1 - O(\epsilon))$-approximate maximum matching of the core must include a constant fraction of the special edges going from $A_k$ to $A_k$. Our goal is to hide these special edges and show that finding each one of them requires at least $d^{k-o(1)}$ queries to the graph. To do this, it is important {\em not} to give away the layer of a vertex. Towards this, our first idea is to assign a random ID to each of the vertices of the core and sort the adjacency lists randomly.

The nice thing about the core is that the local neighborhoods of all the vertices in higher levels are symmetric. In particular, it is not possible to distinguish an $A_k$ vertex $v$ from a $B_k$ vertex without reaching an $S$ vertex in its neighborhood, which are all at distance at least $2k$ from $v$. Note that while there are indeed $\Omega(d^{k})$ vertices in the $2k$-hop of a vertex $v \in A_k$, an LCA is not obligated to explore the whole $2k$-hop of $v$. In fact, a random walk starting from any vertex $v$ reaches an $S$ vertex in just $O_\epsilon(d)$ steps in expectation. Moreover, the distribution of the length of such a random walk until reaching $S$ (which can be approximated sufficiently well with some $O_\epsilon(\log n)$ repetitions) is enough to determine the layer of its starting vertex correctly with high probability. Therefore, we need more ideas to hide the layers of the core.

\paragraph{Step 2 --- Delusive Vertices:} Delusive vertices are a key component of our construction. Their main purpose is to guarantee what we showed the core cannot do on its own: hiding its layers. In our final construction, we will have a hierarchy of delusive vertices. But let us start with one level and see how it helps. We add a set $D$ of $\Theta(\frac{\epsilon}{1+\epsilon} n)$ delusive vertices to the graph. We connect every vertex in $\{A_i, B_i\}_{i \in [k]}$ to $\epsilon d$ delusive vertices in $D$.\footnote{We note that after connecting the $D$ vertices to all of $A_i, B_i$, the resulting graph will no longer be bipartite. Minor modifications will be needed to convert the graph into a bipartite one.} This can be done in a way such that all the $A_i, B_i, D$ vertices have the same degree $d' = d + \epsilon d + 1$ overall, and each vertex in $D$ has the same number of edges to all of the $A_i, B_i$ layers.

It turns out that adding these delusive vertices is enough to kill the random-walk based algorithm outlined above. Indeed, because $\epsilon$ fraction of neighbors of each $A_i, B_i$ vertex goes to $D$, the random walk is expected to hit $D$ every $\Theta(1/\epsilon)$ steps. As this is much smaller than the $\Omega(d)$ expected steps to hit an $S$ vertex, the random walk, w.h.p., sees a $D$ vertex before reaching $S$. On the other hand, the moment that we hit $D$, we completely lose information about where the random walk started. This is because conditioned on having reached a delusive vertex $u \in D$, all the layers have the same probability of being $u$'s predecessor in the walk as $u$ has the same degrees to all the layers.

While one layer of delusive vertices kills the random walk algorithm, it does not yet imply that $d^{k-o(1)}$ queries are needed for determining the label of $A_k$ vertices. In fact, it is still possible to determine the label of any vertex in just $\widetilde{O}(d^2)$ time! To see this, observe first that it is possible to determine whether a vertex is a $B_1$ vertex in $O(d)$ time by simply scanning its neighbors and checking whether there is an $S$ vertex among them. Now suppose that our task is to determine whether a vertex $v$ belongs to $D$. Since only the vertices in $D$ have $\epsilon$ fraction of their neighbors in $B_1$, we can random sample $\widetilde{O}(1)$ neighbors of $v$, check which ones belong to $B_1$, and report  $v \in D$ iff this fraction is sufficiently close to $\epsilon$. Now that we can check if a vertex belongs to $D$ in $\widetilde{O}(d)$ time, we can modify the random walk algorithm, ensuring that we never step on a $D$ vertex by running this test on each vertex that it visits. This only multiplies the running time of the random walk algorithm by a $\widetilde{O}(d)$ factor, thus it takes $\widetilde{O}(d^2)$ time to determine the core layers with one level of delusive vertices.

\paragraph{Step 3 --- A Hierarchy of Delusive Vertices:} In our final construction, instead of just a single layer of delusive vertices, we have a hierarchy of $k = \Theta(1/\epsilon)$ levels of delusive vertices $D_1, \ldots, D_k$. We ensure that the total number of vertices in $D_1, \ldots, D_k$ is $O(\epsilon^2 n)$ so that adding them to the graph does not drastically change the maximum matching of the core. As illustrated in \cref{fig:delusives}, for any $i$, vertices in $D_i$ are made adjacent to $A_j, B_j$ for all $j \geq i$ and to all $D_j$ for $j > i$. Intuitively, while we can still check whether $v \in D_1$ in $\widetilde{O}(d)$ time by examining what fraction of its neighbors belongs to $B_1$, the same cannot be done for $D_2, D_3, \ldots$ as they do not have any direct neighbors in $B_1$. In particular, determining whether a vertex $v$ belongs to $D_i$ (or even $A_i, B_i$) will require $d^{i - o(1)}$ queries in the neighborhood of $v$ which effectively hides the core layers.

\begin{figure*}
    \centering
    \includegraphics[width=1\textwidth]{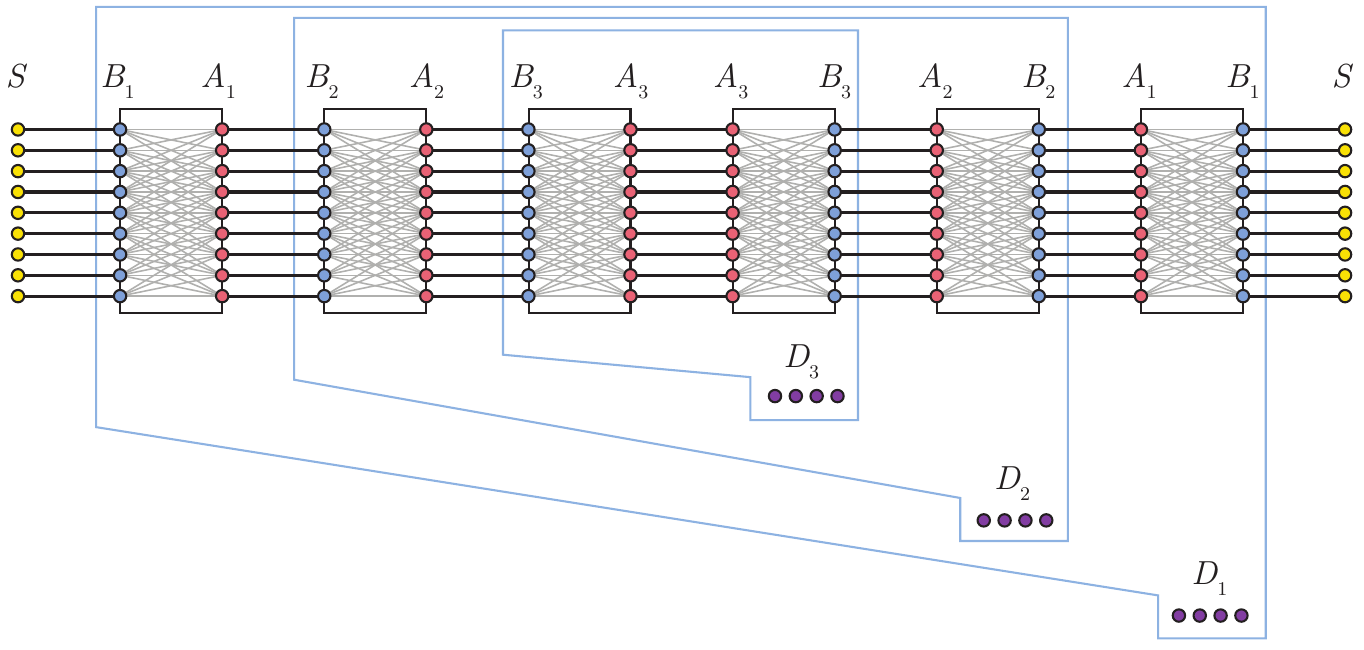}
    \caption{This figure shows how the $k$ levels of delusive vertices are made adjacent to the core. For simplicity, this figure does not show the edges of the delusive vertices, but all the edges of each $D_i$ vertex goes to the vertices in the smallest blue box enclosing it.}
    \label{fig:delusives}
\end{figure*}

\paragraph{Step 4 --- Binomial Degrees:} The $4^{\text{th}}$ and last step of our construction is more of a technical modification to the construction discussed above that is important for our proofs. In the graph illustrated above, each vertex has a fixed number of edges to every layer. Take a vertex $v \in B_1$ for example. It has one neighbor in $S$, $d$ neighbors in $A_1$, and $\epsilon d$ neighbors in $D_1$. In our final construction, we want every neighbor of $v \in B_1$ to belong to $A_1, D_1, S$ independently from the rest of neighbors of $v$. To achieve this, we first draw the number of edges of $v$ to each of $A_1, D_1, S$ from a suitable binomial distribution with the right expected value and then try to satisfy these drawn degrees. A challenge that arises is that the drawn degree sequences of all vertices might not be realizable simultaneously. For instance, if the sum of degrees of $B_1$ to $D_1$ is not the same as the sum of degrees from $D_1$ to $B_1$, then clearly the graph is not realizable. Nonetheless, we show that by modifying the drawn degrees of a small number of ``broken vertices'', the resulting degree sequence will be realizable using a theorem of Gale-Ryser (see \cref{pro:gale–ryser}). We also show that the algorithm will, w.h.p., never see a broken vertex. Effectively, this implies that the layers of the neighbors of any vertex that the algorithm sees will be independent.

\subsection{Formalizing the Lower Bound: The Label Guessing Game on Trees}\label{sec:techniques-label-guessing}

Up to this point, we have presented a high-level overview of our input graph and have also explained why a certain random-walk based algorithm cannot find a $(1, \epsilon n)$-approximate matching of it with less than $d^{\Omega(1/\epsilon)}$ queries. In this section, we overview how we prove this lower bound against all algorithms.

\paragraph{The Label Guessing Game on Trees:} We reduce our lower bound to a clean ``label guessing game'' on a Markovian tree (see \cref{fig:tree}). In this problem, we have a tree $T$ which initially only involves a single vertex $v$ that is going to be the root of $T$ throughout. At each step, the algorithm can adaptively pick a vertex $u \in T$ of its choice. Doing so will add a direct child below $u$. Each vertex added to $T$ will have a {\em hidden} label. The goal is to guess the label of the root vertex $v$ while querying a few vertices in its subtree. The hidden labels correspond to the vertex subsets of our input distribution. That is, each vertex has one label that is either $S$ or $A_i, B_i, D_i$ for some $i \in [k]$. The labels of the children of each vertex $u$ are drawn independently from a distribution that depends only on the label of their parent $u$. These transition probabilities come from our input distribution. For example, each $A_1$ vertex in our input graph has $d$ expected neighbors in $B_1$, $\epsilon d$ expected neighbors in $D_1$, and $\widetilde{O}(1)$ expected neighbors in $B_2$. Thus, once we open a child $w$ for a vertex $u$ whose hidden label is $A_1$, its child $w$ takes label $B_1$ with probability $1-\Theta(\epsilon)$, label $D_1$ with probability $\Theta(\epsilon)$, and takes label $B_2$ with probability $\widetilde{\Theta}(1/d)$ independently. The only information that the algorithm is given is whether the label of each vertex in the tree is $S$ or not. \cref{fig:tree} shows an instance of the label guessing game and two of its possible realizations.

\begin{figure*}
    \centering
    \includegraphics[width=1\textwidth]{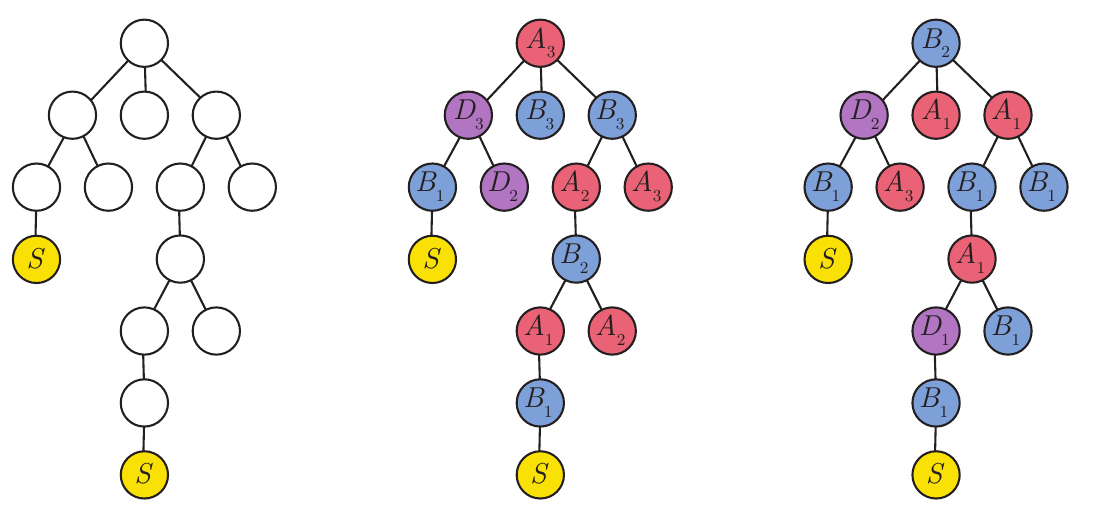}
    \caption{An example of the label guessing game. The tree on the left is what the algorithm sees. In particular, all the labels except for the $S$ labels are hidden from the algorithm. On the right, we have two possible realizations of the labels leading to the same observed tree. The algorithm must pick its queries in such a way that it can guess the label of the root.}
    \label{fig:tree}
\end{figure*}

\paragraph{The Reduction to the Label Guessing Game:} We show that any LCA for $(1, \epsilon n)$-approximate matching for our input construction leads to an efficient label guessing algorithm in the tree model. To show this, we prove that any LCA that queries $d^{O(1/\epsilon)}$ entries of the graph, with high probability, only sees a (rooted) forest. The proof 
relies heavily on the fact that the edges of the input graph $G$ are sufficiently random (even conditioned on satisfying the degree constraints and conditioned on the previous $d^{O(1/\epsilon)}$ queries) and thus expand well. Once we prove this, we are immediately done: conditioned on the high probability event that the LCA does not discover a cycle, the problem becomes exactly the same as the label guessing game.

\paragraph{Lower Bounds for the Label Guessing Game:} Lower bounding the number of queries needed to solve the label guess game is the crux of our analysis. Our proof consists of two parts. In the first part of the proof, we show that any algorithm that solves the label guessing game must find a path from the root to an $S$ vertex that does not go through a certain subset of delusive vertices that we call mixer vertices (\cref{def:mixer-vertex}). To formalize this, via a careful coupling argument, we show that if every path from the root to an $S$ vertex contains a mixer vertex, then the label of the root is equally likely to be, say, $A_k$ or $B_k$. In the second part of the proof, we prove that to discover a path from root of level $k$ to $S$ that does not contain any mixer vertex, the subtree below the root must include at least $d^{k- o(1)}$ vertices. The proof of this is close (but more general) than the arguments we discussed above for why the random-walk based algorithm does not work.


\section{Preliminaries}\label{sec:prelim}

\paragraph{Notation:} In this paper, we let $G = (V, E)$ be the input graph, $n$ to be the number of vertices, $\Delta$ to be the maximum degree of the graph, and $Q$ be the number of queries that the algorithm makes. Moreover, we use $\widetilde{O}(1)$ to hide $\polylog Q$ factors.

\paragraph{Probabilistic tools:} We use the following standard form of Chernoff bound in our paper.

\begin{proposition}[Chernoff Bound]\label{prop:chernoff}
    Let $X_1, X_2, \ldots, X_k$ be independent Bernoulli random variables, and let $X = \sum_{i=1}^{k} X_i$. Then, for any $\delta > 0$, 
    $\Pr[|\E[X] - X| \geq \delta] \leq 2\exp\left(-\frac{\delta^2}{3\E[X]}\right).$
\end{proposition}

\paragraph{Graph theory definitions/tools:}  We use $\mu(G)$ to denote the size of the maximum matching of graph $G$, $\nu(G)$ to denote the size of the vertex cover of graph $G$.

\begin{proposition}[K\"{o}nig’s Theorem]\label{prop:konig}
For any bipartite graph $G$, it holds $\mu(G) = \nu(G)$.
\end{proposition} 

\paragraph{Bigraphic pairs of sequences:}  We use the following results on {\em bigraphic pairs of sequences} defined below.

\begin{definition}[Bigraphic Pairs of Sequences]\label{def:bigarphic}
Let $a = (a_1, a_2, \ldots a_n)$ and $b = (b_1, b_2, \ldots, b_m)$ be two sequences of non-negative integers. We say this is a bigraphic pair of sequences if there exists a bipartite graph where $a$ corresponds to the degree sequence of one part of the graph and $b$ corresponds to the degree sequence of the other part.
\end{definition}

\begin{proposition}[Gale–Ryser Theorem]\label{pro:gale–ryser}
    Let $(a_1, a_2, \ldots a_n)$ and $(b_1, b_2, \ldots, b_m)$ be two sequences of non-negative integers such that $a_1 \geq a_2 \geq \ldots \geq a_n$. Then, these two sequences are bigraphic if and only if $\sum_{i=1}^n a_i = \sum_{i=1}^m b_i$, and
    \begin{align*}
        \sum_{i=1}^r a_i \leq \sum_{i=1}^m \min(b_i, r) \qquad \text{for all $1\leq r\leq n$.}
    \end{align*}
\end{proposition}
\section{Input Distribution and its Characteristics}\label{sec:input-distribution}

In this section, we describe the input distribution of our construction. We have two types of input graphs where the first graph has an almost perfect matching and for the second graph, only $(1-\epsilon)$ fraction of vertices are matched in the maximum matching. We prove that any {\em deterministic} LCA which with probability at least 0.51 computes a $(1, \epsilon n)$-approximate maximum matching of graphs drawn from this distribution, must spend at least $\Delta^{\Omega(1/\epsilon)}$ time. From Yao's minimax theorem \cite{Yao77}, we thus get that any {\em randomized} LCA that computes a $(1, \epsilon n)$-approximate matching for all inputs with success probability at least 0.51 must also spend at least $\Delta^{\Omega(1/\epsilon)}$ time per query.

Let $N$ be a parameter that controls the number of vertices in our input distribution. Moreover, in our construction, let $d \leq n^{\epsilon/3}$ be a parameter that controls the degree of vertices. Graphs in our input distribution have $n = (1/2 + 1/\epsilon + \epsilon - \epsilon^2 / 2)N$ vertices. We first describe the vertex set of the graphs in our distribution.

\paragraph{The vertex set:} The vertex set consists of disjoint subsets $A^1_i, B^1_i, A^2_i, B^2_i$ for each $i \in [1/\epsilon]$  as well as two subsets $S^1$ and $S^2$. Each of these subsets except $A^1_{1/\epsilon}, A^2_{1/\epsilon}, S^1$, and $S^2$, has exactly $N/4$ vertices. Each of $A^1_{1/\epsilon}$ and $A^2_{1/\epsilon}$ has $(1-\epsilon^2 )N/4$ vertices. Also, each of $S^1$ and $S^2$ has $N/4$ vertices. Moreover, the vertex set consists of subsets $D_i$ of {\em delusive vertices} for $i \in [1/\epsilon]$, where each of these $1/\epsilon$ subsets has exactly $\epsilon^2 N$ vertices. So we have
\begin{align*}
    |A^1_i| = |A^2_i| = |B^1_i| = |B^2_i| = \frac{N}{4} \qquad \forall i \in [\frac{1}{\epsilon} - 1],
\end{align*}
\vspace{-0.5cm}
\begin{align*}
    |B^1_{1/\epsilon}| = |B^2_{1/\epsilon}| = \frac{N}{4}, \quad |A^1_{1/ \epsilon}| = |A^2_{1/ \epsilon}| = (1-\epsilon^2)\frac{N}{4},
\end{align*}
\vspace{-0.5cm}
\begin{align*}
    |S^1| = |S^2| = \frac{N}{4},
\end{align*}
\vspace{-0.5cm}
\begin{align*}
    |D_i| = \epsilon^2 N \quad \forall i \in [\frac{1}{\epsilon}].
\end{align*}
Hence, the total number of vertices in each graph of our input distribution is $n = (1/2 + 1/\epsilon + \epsilon - \epsilon^2 / 2)N$.

\paragraph{The edge set:} For the edge set, we have two different distributions; \yesdist{} and \nodist{}. In \yesdist{}, the graph has an almost perfect matching. On the flip side, a maximum matching of \nodist{} leaves at least $\epsilon n$ vertices unmatched. In our input distribution, we draw the graph from \yesdist{} with probability $1/2$ and from \nodist{} with probability $1/2$.

Let $X, Y \in \bigcup_{i=1}^{1/\epsilon} \{A^1_i, A^2_i, B^1_i,B^2_i, D_i \} \cup \{S^1, S^2\}$ be any two vertex subsets. We use $\deg_X^Y(v)$ to denote the number of vertices of subset $Y$ that are adjacent to a single vertex $v \in X$. First, we show how the degree of vertices will be determined in \yesdist{} and \nodist{}, then we describe how to construct a graph with the corresponding degree sequence. Each vertex except vertices of $S$ has exactly $d' = d + \epsilon^3 d + \log^4 N$ neighbors. Also, all vertices of $S$ have $\log^4 N$ neighbors. For a vertex $u \in X$, the type of its neighbor $v \in Y$ is determined independently at random according to the following binomial distribution for both \yesdist{} and \nodist{} (it helps to recall \cref{fig:delusives} of \cref{sec:techniques}):

\begin{itemize}
    \item Vertices of $S^j$ have $\log^4 N$ neighbors and the neighbors only can be $B^j_1$ for $j \in \{1, 2\}$.
    \item If $X = B^j_1$ for $j \in \{1, 2\}$:
    \begin{align*}
        \Pr[Y = S^j] = \frac{\log^4 N}{d'}, \quad \Pr[Y = A^j_1] = \frac{d}{d'},
    \end{align*}
    \vspace{-0.5cm}
    \begin{align*}
        \Pr[Y = D_1] = \frac{\epsilon^3 d}{d'},
    \end{align*}
    \item If $X = B^j_i$ for $j\in \{1, 2\}$ and $1 < i < 1/\epsilon$:
    \begin{align*}
        \Pr[Y = A^j_{i-1}] = \frac{\log^4 N}{d'}, \quad &\Pr[Y = A^j_i] = \frac{d}{d'},
    \end{align*}
        \vspace{-0.5cm}
    \begin{align*}
        \Pr[Y = D_i] = \frac{(1/\epsilon - i + 1) \epsilon^4 d}{d'},
    \end{align*}
        \vspace{-0.5cm}
    \begin{align*}
        \Pr[Y = D_k] = \frac{\epsilon^4 d}{d'} \quad \text{for $k < i$}.
    \end{align*}

    \item If $X = A^j_i$ for $j \in \{1, 2\}$ and $1 \leq i < 1/\epsilon$:
    \begin{align*}
        \Pr[Y = B^j_{i+1}] = \frac{\log^4 N}{d'}, \quad &\Pr[Y = B^j_i] = \frac{d}{d'},
    \end{align*}
    \vspace{-0.5cm}
    \begin{align*}
        \Pr[Y = D_i] = \frac{(1/\epsilon - i + 1) \epsilon^4 d}{d'},
    \end{align*}
    \vspace{-0.5cm}
    \begin{align*}
        \Pr[Y = D_k] = \frac{\epsilon^4 d}{d'} \quad \text{for $k < i$}.
    \end{align*}
    \item If $X = D_i$ for $i \in [1/\epsilon - 1]$:
    \begin{align*}
        &\Pr[Y = D_i] = \frac{(1-2\epsilon + 2i\epsilon^2 - 5\epsilon^2 / 2 + 3\epsilon^4)d + \log^4 N}{d'},
    \end{align*}
    \vspace{-0.5cm}
    \begin{align*}
        \Pr[Y = D_j] = \frac{\epsilon^4 d}{d'} \quad \text{for $j \neq i$},
    \end{align*}
    \vspace{-0.5cm}
    \begin{align*}
        \Pr[Y = A^j_{1/\epsilon}] = \frac{(\epsilon^2 - \epsilon^4)d}{d'}  \quad \text{for $j \neq i$},
    \end{align*}
    \vspace{-0.5cm}
    \begin{align*}
        \Pr[Y = A^j_i] = \Pr[Y = B^j_i] &= \frac{(1/\epsilon - i + 1)\cdot \epsilon^2 d/4}{d'}\\ 
        &\text{for $j \in \{1, 2\}$},
    \end{align*}
    \vspace{-0.5cm}
    \begin{align*}
        \Pr[Y = A^j_k] &= \Pr[Y = B^j_k] = \Pr[Y = B^j_{1/\epsilon}] = \frac{\epsilon^2 d / 4}{d'}\\ &\text{for $j \in \{1, 2\}$ and $i < k < 1/\epsilon$}.
    \end{align*}
    \item If $X = D_i$ for $i = 1/\epsilon$:
    \begin{align*}
        &\Pr[Y = D_i] = \frac{(1 - 5\epsilon^2 / 2 + 3\epsilon^4)d + \log^4 N}{d'},
    \end{align*}
    \vspace{-0.5cm}
    \begin{align*}
        \Pr[Y = D_j] = \frac{\epsilon^4 d}{d'} \quad \text{for $j \neq i$},
    \end{align*}
    \vspace{-0.5cm}
    \begin{align*}
        \Pr[Y = A^j_{1/\epsilon}] = \frac{(\epsilon^2 - \epsilon^4)d}{d'} \quad \text{for $j \in \{1, 2\}$},
    \end{align*}
    \vspace{-0.5cm}
    \begin{align*}
        \Pr[Y = B^j_{1/\epsilon}] = \frac{\epsilon^2 d/4}{d'} \quad \text{for $j \in \{1, 2\}$}.
    \end{align*}
\end{itemize}

Distribution of neighbors of vertices in $A^j_{1/\epsilon}$ and $B^j_{1/\epsilon}$ for $j \in \{1, 2\}$ is different in \yesdist{} and \nodist{}. The following binomial distribution is the distribution of neighbors in \yesdist{}:

\begin{itemize}
        \item If $X = B^j_{1/\epsilon}$ for $j \in \{1, 2\}$:
    \begin{align*}
        \Pr[Y = A^j_{1/\epsilon-1}] = \frac{\log^4 N}{d'}, \quad \Pr[Y = A^j_{1/\epsilon}] = \frac{(1-\epsilon^2)d}{d'},
    \end{align*}
        \vspace{-0.5cm}
    \begin{align*}
        \Pr[Y = B^{3 - j}_{1/\epsilon}] = \frac{\epsilon^2 d}{d'},
    \end{align*}
        \vspace{-0.5cm}
    \begin{align*}
        \Pr[Y = D_k] = \frac{\epsilon^4 d}{d'} \quad \text{for $k \leq 1/\epsilon$}.
    \end{align*}
    \item If $X = A^j_{1/\epsilon}$ for $j \in \{1, 2\}$:
    \begin{align*}
        \Pr[Y = A^{3-j}_{1/\epsilon}] = \frac{\log^4 N}{d'}, \quad \Pr[Y = B^j_{1/\epsilon}] = \frac{d}{d'},
    \end{align*}
            \vspace{-0.5cm}
    \begin{align*}
        \Pr[Y = D_k] = \frac{\epsilon^4 d}{d'} \quad \text{for $k \leq 1/\epsilon$}.
    \end{align*}
    \end{itemize}

        The following binomial distribution is the distribution of neighbors in \nodist{}:
    \begin{itemize}
        \item If $X = B^j_{1/\epsilon}$ for $j \in \{1, 2\}$:
    \begin{align*}
        &\Pr[Y = A^j_{1/\epsilon-1}] = \frac{\log^4 N}{d'},
    \end{align*}
    \vspace{-0.5cm}
    \begin{align*}
        \Pr[Y = A^j_{1/\epsilon}] = \frac{(1-\epsilon^2)(d+\log^4 N)}{d'},
    \end{align*}
    \vspace{-0.5cm}
    \begin{align*}
        \Pr[Y = B^{3 - j}_{1/\epsilon}] = \frac{\epsilon^2 (d+\log^4 N) - \log^4 N}{d'},
    \end{align*}
    \vspace{-0.5cm}
    \begin{align*}
        \Pr[Y = D_k] = \frac{\epsilon^4 d}{d'} \quad \text{for $k \leq 1/\epsilon$}.
    \end{align*}
    \item If $X = A^j_{1/\epsilon}$ for $j \in \{1, 2\}$:
    \begin{align*}
       \Pr[Y = B^j_{1/\epsilon}] = \frac{d+ \log^4 N}{d'},
    \end{align*}
    \vspace{-0.5cm}
    \begin{align*}
        \Pr[Y = D_k] = \frac{\epsilon^4 d}{d'} \quad \text{for $k \leq 1/\epsilon$}.
    \end{align*}
\end{itemize}

Note that for two different subsets $X$ and $Y$, we only described how to determine $\deg_X^Y(v)$. Unfortunately, it may not be possible to construct a graph with the resulting degree sequence. As described in \cref{sec:techniques}, if the sum of degrees from $A^1_1$ to $B^1_1$ is different from the sum of degrees from $B^1_1$ to $A^1_1$, then no graph can satisfy this degree sequence. Nonetheless, we prove that by ignoring the degrees of at most $O(\sqrt{nd}\cdot \log n)$ vertices, which we call {\em broken vertices}, the degrees of the rest of the vertices can be satisfied with high probability. The following \cref{lem:convert-degree-to-edge} is useful in showing how we add the edges according to the degree sequence.

\begin{definition}[Broken Vertices]
    Take a vertex $v$ in our input graph. We say $v$ is a broken vertex if its degree in the final graph is different from the degree initially drawn from the binomial distribution.
\end{definition}

\begin{lemma}\label{lem:convert-degree-to-edge}
Let $a_1 \geq a_2 \geq \ldots \geq a_{k_1}$ and $b_1 \geq b_2 \geq \ldots \geq b_{k_2}$ be two sequences of non-negative integers where $a_i$ is drawn from a Binomial distribution with $\eta_1$ trials and success probability $\rho_1$, and $b_i$ is drawn from a Binomial distribution with $\eta_2$ trials and success probability $\rho_2$ for all $i$, and suppose $\sum_i^{k_1}a_i \geq \sum_i^{k_2}b_i$. Also, assume that $k_1\eta_1 \rho_1 = k_2\eta_2 \rho_2$, $\eta_1 \rho_1 = \Omega(\log^4 n)$, $\eta_2 \rho_2 = \Omega(\log^4 n)$, $\eta_1 \rho_1 = O(d)$, $\eta_2 \rho_2 = O(d)$, $k_1 = \Theta(n)$, and $k_2 = \Theta(n)$. Then, with high probability, there exists a sequence of non-negative integers $a'_1 \geq a'_2 \geq \ldots \geq a'_{k_1}$ such that all the following hold:
\begin{itemize}
    \item $0 \leq a_i - a'_i \leq 10\sqrt{\eta_1 \rho_1} \log n$ for all $i$,
    \item $(a'_1, a'_2, \ldots, a'_{k_1})$ and $(b_1, b_2, \ldots, b_{k_2})$ is a bigraphic pair of sequences (see \cref{def:bigarphic}), and
    \item there are at most $O(\sqrt{k_1 \eta_1 \rho_1}\cdot \log n)$ elements in the sequence where $a_i \neq a_i'$.
\end{itemize}
\end{lemma}
\begin{proof}
    First, note that since both sequences are drawn from a binomial distribution, by applying a Chernoff bound, with  a probability of at least $1 - 2n^{-5}$, we have $a_i \in (\eta_1 \rho_1 \pm 5\sqrt{\eta_1 \rho_1}\cdot \log n)$ (resp., $b_i \in (\eta_2 \rho_2 \pm 5\sqrt{\eta_2 \rho_2}\cdot \log n)$). Thus, using a union bound, with a high probability this event holds all for $a_i$'s and $b_i$'s.

    Similarly, using the Chernoff bound, we get that with high probability,
    \begin{align*}
    \sum_i^{k_1}a_i \in \left(k_1\eta_1 \rho_1 \pm O(\sqrt{k_1\eta_1 \rho_1}\cdot \log n)\right),\\
    \sum_i^{k_2}b_i \in \left(k_2\eta_2 \rho_2 \pm O(\sqrt{k_2\eta_2 \rho_2}\cdot \log n)\right).
    \end{align*}

    Let $D = \sum_i^{k_1}a_i - \sum_i^{k_2}b_i$. Since $k_1\eta_1 \rho_1 = k_2\eta_2 \rho_2$, by the above bounds, $D \leq O(\sqrt{k_1\eta_1 \rho_1}\cdot \log n)$. We construct a degree sequence $a' = (a'_1, a'_2, \ldots, a'_{k_1})$ in $D$ iterations. Initially, we set $a'_i = a_i$ for all $i$. At each iteration, we choose the maximum $a'_i$ and reduce its value by one. In the end, we sort $a'$ in decreasing order. Note that according to the construction, we have $\sum_i^{k_1}a_i - \sum_i^{k_1}a'_i = D \leq O(\sqrt{k_1\eta_1 \rho_1}\cdot \log n)$. Furthermore, the maximum of $a'$ cannot be less than $\eta_1 \rho_1 - 5\sqrt{\eta_1 \rho_1}\cdot \log n$, as otherwise, $\sum_i^{k_1}a_i$ should be significantly less than $k_1\eta_1 \rho_1 - O(\sqrt{k_1\eta_1 \rho_1}\cdot \log n)$ which is a contradiction. Thus, $0 \leq a_i - a'_i \leq 10\sqrt{\eta_1 \rho_1}\cdot \log n$ for all $1 \leq i \leq k_1$. Therefore, it remains to show that $(a'_1, a'_2, \ldots, a'_{k_1})$ and $(b_1, b_2, \ldots, b_{k_2})$ is a bigraphic pair of sequences.

    For this aim, we use Gale–Ryser theorem in \Cref{pro:gale–ryser}. We need to show that the conditions in this theorem hold for the pair of sequences. Formally, for each $1 \leq r \leq k_1$, we claim that $\sum_i^{r} a'_i \leq \sum_i^{k_2} \min(b_i, r)$. If $r \geq \eta_2\rho_2 + 5\sqrt{\eta_2\rho_2}\cdot \log n$, then
\begin{align*}
    \sum_i^{k_2} \min(b_i, r) = \sum_i^{k_2} b_i \geq \sum_i^{r} a'_i,
\end{align*}
where the first equality follows by the high probability event of having $b_i \leq \eta_2\rho_2 + 5\sqrt{\eta_2\rho_2}\cdot \log n$ for all $1 \leq i \leq k_2$. If $r < \eta_2\rho_2 + 5\sqrt{\eta_2\rho_2}\cdot \log n$, then
\begin{align*}
    \sum_i^{r} a'_i &\leq r \cdot (\eta_1 \rho_1 + 5\sqrt{\eta_1 \rho_1}\cdot \log n)\\
    &\leq O(d^2) \leq n \leq \sum_i^{k_2} \min(b_i, r),
\end{align*}
which completes the proof.
\end{proof}

\begin{corollary}\label{cor:new-degree}
Let $(a_1, a_2, \ldots, a_n)$ be the degree sequence that is produced by the construction. Then, there exists a graph with sequences $(b_1, b_2, \ldots, b_n)$ such that there exists at most $O(\sqrt{n d} \log n)$ broken vertices in the constructed graph.
\end{corollary}
\begin{proof}
    Proof follows by applying \Cref{lem:convert-degree-to-edge} for the degree sequence of induced subgraph for all pairs $(X,Y)$ such that $X, Y \in (S^1 \cup S^2) \cup (\bigcup_{i=1}^{1/\epsilon} \{A^1_i, A^2_i, B^1_i, B^2_i, D_i\})$.
\end{proof}

\begin{remark}\label{rem:bipartite-note}
    Note that there are edges inside $D_i$ for each $i \in [1/\epsilon]$, hence, we cannot use \Cref{lem:convert-degree-to-edge} to put edges in $G[D_i]$ since the graph is not bipartite. However, we can assume that the number of vertices in each $D_i$ is even, and there are two parts in each $D_i$ where vertices of each part are only connected to the other part. With this small modification, we can use \Cref{lem:convert-degree-to-edge} for $G[D_i]$.
\end{remark}

\paragraph{Edges of the graph:} We use \Cref{cor:new-degree} to construct a graph with the given degree sequence that we determined before. By \Cref{cor:new-degree}, there are at most $O(\sqrt{nd}\log n)$ broken vertices. Distribution \yesdist{} (resp., \nodist{}) picks a graph uniformly from the set of all possible graphs that satisfy the modified degree sequence corresponding to \yesdist{} (resp., \nodist{}).

Now we observe some properties of the input distribution that are immediately implied by the construction and important for the proof.

\begin{observation}\label{obs:broken-vertex-upperbound}
    For any graph that is drawn from the input distribution, with high probability, there exists at most $O(\sqrt{nd} \log n)$ broken vertices.
\end{observation}

\begin{proof}
    The proof follows by \Cref{cor:new-degree}.
\end{proof}

\begin{claim}\label{clm:perfect-matching-whp}
    With high probability, all the following hold:
    \begin{enumerate}
        \item There exists a matching of size $(1-\epsilon^3)N/4$ between vertices of $S^j$ and $B^j_1$ for all $j \in \{1,2\}$ for all $j \in \{1,2\}$.
        \item There exists a matching of size $(1-\epsilon^3)N/4$ between vertices of $A^j_i$ and $B^j_{i+1}$ for all $i \in [1/\epsilon - 1]$ and $j \in \{1,2\}$.
        \item If the input graph is drawn from \yesdist{}, then there exists a matching of size $(1-2\epsilon^2)N/4$ between $A^1_{1/\epsilon}$ and $A^2_{1/\epsilon}$.
    \end{enumerate}
\end{claim}

\begin{proof}
    Let $v \in A_i^j \cup B_{i+1}^j$. Degree of vertex $v$ in $G[A_i^j, B_{i+1}^j]$ is concentrated around $\log^4 N$ with $10\log^3 N$ error since the expected degree is $\log^4 N$, we can show that using a standard Chernoff bound, the error is at most $10\log^3 N$ with high probability. Furthermore, by \Cref{lem:convert-degree-to-edge}, the degree of a vertex can decrease by $10\log^3 N$ additive value when we put edges in the graph using \Cref{lem:convert-degree-to-edge}. Thus, the degree cannot be smaller than $\log^4 N - 20\log^3 N$ and larger than $\log^4 N + 10 \log^3 N$. We construct a fractional matching such that for each edge $e$ in $G[A_i^j, B_{i+1}^j]$, we set $f_e = 1/(\log^4 N + 10\log^3 N)$. Since the degree is at most $\log^4 N + 10 \log^3 N$, this fractional matching is feasible. Let $E(v)$ be the set of edges incident to $v$ in $G[A_i^j, B_{i+1}^j]$, then due to the integrality gap of the fractional matching polytope in bipartite graphs, we have 
    \begin{align*}
    \mu(G[A_i^j, B_{i+1}^j]) \geq \sum_{v \in A_i^j} \sum_{e \in E(v)} f_e &\geq \sum_{v \in A_i^j} \frac{\log^4 N - 20\log^3 N}{\log^4 N + 10 \log^3 N}\\
    &\geq \frac{N}{4} \cdot \left(1 - \frac{40}{\log N}\right)\\
    &\geq (1-\epsilon^3)\frac{N}{4},
    \end{align*}
    concluding the proof for statement (2).\footnote{In the proof of this lemma, we need $\epsilon$ to be constant. However, we might use a slightly modified version of the result by \cite{friezepittel} to show that there exists a perfect matching in $G[A_i^j, B^j_{i+1}]$ and $G[A_{1/\epsilon}^1, A_{1/\epsilon}^2]$. With this change, we do not need the assumption for $\epsilon$ to be constant.} A similar argument also works for statement (1) since the degrees and sizes of subgraphs are the same.

    Proof of the third statement is similar to the second statement since the degree of vertices in $G[A^1_{1/\epsilon}, A^2_{1/\epsilon}]$ is concentrated around $\log^4 N$ with $10 \log^3 N$ error. Let $E(v)$ be the set of edges incident to $v$ in $G[A^1_{1/\epsilon}, A^2_{1/\epsilon}]$. If we construct the same fractional matching, then we get
    \begin{align*}
    \mu(G[A^1_{1/\epsilon}, A^2_{1/\epsilon}]) &\geq \sum_{v \in A^1_{1/\epsilon}} \sum_{e \in E(v)} f_e \\
    &\geq \sum_{v \in A^1_{1/\epsilon}} \frac{\log^4 N - 20\log^3 N}{\log^4 N + 10 \log^3 N} \\
    &\geq \frac{(1-\epsilon^2)N}{4} \cdot \left(1 - \frac{40}{\log N}\right)\\
    &\geq (1-2\epsilon^2)\frac{N}{4},
    \end{align*}
    concluding the proof for statement (3).
\end{proof}

\begin{lemma}\label{lem:matching-size}
Let $G_\yes \sim \yesdist$ and $G_\no \sim \nodist$. Then, with high probability,
\begin{itemize}
    \item $\mu(G_\yes) \geq \left(2/\epsilon + 1 - 4\epsilon^2\right)\frac{N}{4}$,
    \item $\mu(G_\no) \leq  \left(2/\epsilon + 4\epsilon\right) \frac{N}{4}$.
\end{itemize}
\end{lemma}

\begin{proof}
Consider the graph $G_\yes$. For each $i\in [1/\epsilon - 1]$ and $j \in \{1,2\}$, by \Cref{clm:perfect-matching-whp}, we have a matching between $A^j_i$ and $B^j_{i+1}$ that matches $(1-\epsilon^3)N/4$ vertices of each part. Also, for each $j \in \{1,2\}$, we have a matching between $S^j$ and $B^j_1$  that matches $(1-\epsilon^3)N/4$ vertices of each part. Moreover, there exists a matching between $A^j_{1/\epsilon}$ and $A^{3-j}_{1/\epsilon}$ that matches $(1-2\epsilon^2)N$ vertices of each part. Since the vertex sets are disjoint, by taking the edges of all these matchings, we have
\begin{align*}
    \mu(G_\yes) &\geq 2 \left(\frac{1}{\epsilon} - 1\right)(1-\epsilon^3)\frac{N}{4} + \frac{N}{2} + (1-2\epsilon^2)\frac{N}{4}\\
    &\geq \left(\frac{2}{\epsilon} + 1 - 3\epsilon^2\right)\frac{N}{4}.
\end{align*}
Now consider $G_\no$. First, we show that $\mu(G_\no[V \setminus \bigcup_{i=1}^{1/\epsilon} D_i]) \leq N/(2\epsilon)$. To see this, note that $G_\no[V \setminus \bigcup_{i=1}^{1/\epsilon} D_i]$ is a bipartite graph which implies that the size of the vertex cover of this graph is equal to the size of the maximum matching by K\"{o}nig’s Theorem (\Cref{prop:konig}). Since there is no edge in the induced graph $G[\bigcup_{i=1, j \in \{1,2\}}^{1/\epsilon} A^j_i \cup \{S^1, S^2\}]$, we take $\bigcup_{i=1, j \in \{1,2\}}^{i \leq 1/\epsilon} B^j_i$ as the vertex cover of this graph. Furthermore, since $|\bigcup_{i=1}^{1/\epsilon} D_i| = \epsilon N$, the number of maximum matching edges that have at least one endpoint in $|\bigcup_{i=1}^{1/\epsilon} D_i|$ is at most $\epsilon N$. Thus, we have 
\begin{align*}
    \mu(G_\no) &\leq \mu\left(G_\no\left[V \setminus \bigcup_{i=1}^{1/\epsilon} D_i\right]\right) + \epsilon N\\ 
    &\leq \frac{N}{2\epsilon} + \epsilon N \\
    &= \left(\frac{2}{\epsilon} + 4\epsilon\right) \frac{N}{4}. \qquad \qedhere
\end{align*}
\end{proof}

\begin{corollary}\label{cor:matching-size-query}
Let $\epsilon < 0.07$. Any algorithm that estimates the size of maximum matching of a graph $G$ that is drawn from input distribution within a factor of $(1, \epsilon n / 7)$ with probability at least 0.51, must be able to distinguish whether $G$ belongs to \yesdist{} or \nodist{}.
\end{corollary}
\begin{proof}
    Note that we have
    \begin{align*}
        \mu(G_\yes) - \mu(G_\no) \geq \left(1 - 4\epsilon - 4\epsilon^2 \right)\frac{N}{4}.
    \end{align*}
    Moreover, since $\epsilon < 0.07$,
    \begin{align*}
        \left(1 - 4\epsilon - 4\epsilon^2 \right)\frac{N}{4} > \frac{N}{6}.
    \end{align*}
    Combining with the fact that $N > 6/7 \cdot (\epsilon n)$, we obtain the claimed bound.
\end{proof}

\section{A Reduction to a Label Guessing Game on Trees}\label{sec:no-cycle}

In this section, we prove that any algorithm that makes $o(d^{1/\epsilon})$ queries, cannot discover any cycle and only sees a rooted forest with high probability. This effectively reduces the problem to the label guessing game on trees that we outlined in \cref{sec:techniques-label-guessing}. The following lemma formalizes the main result of this section.

\begin{lemma}\label{lem:rooted-forest}
    Let $\mathcal{A}$ be any algorithm that makes at most $o(d^{1/\epsilon})$ queries. Let $F_0$ be the empty graph before the algorithm makes any queries, and for $t > 0$, let $F_t$ be the subgraph that $\mathcal{A}$ discovers after $t$ queries. The following property holds throughout the execution of $\mathcal{A}$ with probability $1 - o(1)$: Suppose that the $t$-th query is made to the adjacency list of vertex $u$ and edge $(u, v)$ is returned. Then, vertex $v$ is a singleton vertex in $F_{t-1}$.
\end{lemma}

\begin{remark}
\Cref{lem:rooted-forest} implies that the discovered forest can be thought of as a rooted forest. In other words, if edge $(u, v)$ is discovered by the algorithm at step $t$ and $v$ is the singleton vertex, then $v$ is the leaf of $F_t$.
\end{remark}

The main technical part to prove \Cref{lem:rooted-forest} is to show that at any time during the execution of the algorithm, for any pair of vertices $(u, v)$ that $\mathcal{A}$ has not discovered an edge yet, the probability of having an edge $(u, v)$ is at most $O(d/n)$. To see this, note that if $u$ and $v$ belong to two blocks in the construction that there is no edge between them, then the probability of having an edge between them is zero. Now if they belong to two blocks that we put edges between them, then since we put almost regular graphs with a degree of at most $O(d)$ between any two blocks, the probability of having that edge is $O(d/n)$. This is not a formal argument and in order to formalize this intuition, we use a coupling argument.

\begin{lemma}\label{lem:edge-probability-upper-bound}
    Let $(u, v)$ be a pair of vertices that the algorithm has not discovered an edge between them. Then, the probability of having the edge $(u, v)$ in $G$ is $O(d/n)$.
\end{lemma}

Before proving \Cref{lem:edge-probability-upper-bound}, first we show how we can complete the proof of \Cref{lem:rooted-forest} using \Cref{lem:edge-probability-upper-bound}.

\begin{proof}[Proof of \Cref{lem:rooted-forest}]
    The proof consists of two parts. First, we show that during the execution of the algorithm at any time $t$, if $u$ and $v$ are two non-singleton vertices, then there is no edge between $u$ and $v$. We use induction on $t$ to prove this claim. For $t = 0$ this claim clearly holds. At time $t > 0$, suppose that $\mathcal{A}$ finds an edge $(u, v)$ such that $v$ is a singleton in $F_{t-1}$ (similarly, $u$ can be a singleton vertex). Now we need to show that $v$ does not have any edge to non-singleton vertices in $F_{t-1}$ except $v$. Note that the probability of having an edge between $v$ and any of non-singleton vertices in $F_{t-1}$ is $O(d/n)$. Since $\mathcal{A}$ make at most $o(d^{1/\epsilon})$ queries, there are at most $o(d^{1/\epsilon})$ non-singleton vertices in $F_{t-1}$. Thus, by union bound, the probability of having an edge between $v$ and non-singleton vertices of $F_{t-1}$ is $o(d^{1/\epsilon}) \cdot O(d/n) = o(d^{1/\epsilon + 1}/n)$. Moreover, the induction has $o(d^{1/\epsilon})$ steps since the algorithm makes at most $o(d^{1/\epsilon})$ queries. Therefore, the probability of failure over all steps is at most $o(d^{1/\epsilon}) \cdot o(d^{1/\epsilon + 1}/n) = o(1)$ because $d = n^{\epsilon/3}$.

    Second, we show that if we query the adjacency list of a singleton vertex $u$ and the algorithm discovers edge $(u,v)$, then $v$ is also a singleton vertex. Fix a singleton vertex $u$. By \cref{lem:edge-probability-upper-bound}, the probability of having an edge between $u$ and each of the non-singleton vertices in the forest is $O(d/n)$. Hence, the expected number of edges between $u$ and non-singleton vertices is at most $o(d^{1/\epsilon + 1}/n)$ since there are at most $o(d^{1/\epsilon})$ non-singleton vertices. Furthermore, $u$ has $\Omega(d)$ neighbors according to the construction and the adjacency list of $u$ is randomly permuted which implies that the probability of the first neighbor in the adjacency list to be non-singleton is $o(d^{1/\epsilon}/n)$. Since the algorithm makes at most $o(d^{1/\epsilon})$ queries, the probability of seeing an edge between a singleton vertex and non-singleton vertex when the algorithm queries the singleton vertex's adjacency list is at most $o(d^{1/\epsilon}) \cdot  o(d^{1/\epsilon}/n) = o(1)$ by union bound, which completes the proof.
\end{proof}

\subsection{Proof of~\cref{lem:edge-probability-upper-bound}}

Suppose that $u \in X$ and $v \in Y$, where $X$ and $Y$ show the subset in the construction that $u$ and $v$ belong to. If there is no edge in the construction between two subsets $X$ and $Y$, then the probability of having edge $(u,v)$ is zero. Now we consider two possible scenarios for the types $X$ and $Y$: 1) one of $X$ or $Y$ is of type $S^1$ or $S^2$, 2) none of $X$ or $Y$ is of type $S^1$ or $S^2$.

In the first case, without loss of generality assume that $X \in S^1$ and $Y \in B^1_1$. Let $v \in B^1_1$. According to the binomial distribution of neighbors of $v$, the expected number of $S^1$ neighbors of $v$ is $\log^4 N$. Thus, using the Chernoff bound, the total number of edges between $B^1_1$ and $S^1_1$ is not larger than $\frac{N}{3}\log^4 N$ with high probability, which implies that there are at least $\frac{N}{6}\log^4 N$ vertices of $S^1_1$ that have a degree equal to zero. Now let $\mathcal{G}$ be the set of all graphs in the input distribution that have edge $(u,v)$, and $\hat{\mathcal{G}}$ be the set of all graphs in the input distribution that does not have edge $(u,v)$. For a graph in $\mathcal{G}$, we can remove the edge $(u,v)$ and add edge $(w,v)$ for a vertex $w \in S^1_1$ that has degree zero. Since there exists $O(n \log^4 n)$ such $w$, we can couple the initial graph to $\Omega(n\log^4 n)$ graphs in $\hat{\mathcal{G}}$. On the other hand, each graph of $\hat{\mathcal{G}}$ is coupled to at most $O(\log^4 n)$ graphs in $\mathcal{G}$ since the degree of $v$ is at most $O(\log^4 n)$. Hence, we have $|\mathcal{G}| / |\hat{\mathcal{G}}| \leq O(1/n)$, which concludes the proof for the first case since the number of graphs in the input distribution that have the edge $(u, v)$ is $O(1/n)$ fraction of graphs that does not have the edge $(u, v)$.

For the second case, we use a more complicated coupling argument. Suppose that the expected degree of a vertex in $X$ in the subgraph of $G[X, Y]$ is $d_1$ and the expected degree of a vertex in $Y$ is $d_2$ in $G[X, Y]$. By \Cref{lem:convert-degree-to-edge} and using Chernoff bound, the degree of all vertices $X$ is in the range $d_1 \pm d_1/2$ in subgraph $G[X, Y]$. Similarly, the degree of all vertices $Y$ is in the range $d_2 \pm d_2/2$ in subgraph $G[X, Y]$. We define $\mathcal{G}$ and $\hat{\mathcal{G}}$ similar to the previous case. The key idea for this case is that if edge $(u,v)$ exists in a graph, we can find many edges $(x, y)$ such that $x \in X$, $y \in Y$, edge $(x,y)$ is not discovered by the algorithm, and there exist exactly two edges $(u,v)$ and $(x, y)$ in $G[\{u,v,x,y\}]$. Then, by removing edges $\{(u,v), (x,y)\}$ and adding edges $\{(u,y), (x,v)\}$ we can obtain a graph that does not have edge $(u, v)$, its degree sequence does not change, and satisfy all properties of input distribution (if the initial graph is in \yesdist{}, the final graph is also in \yesdist{}. The same statement hold for \nodist{}).

Suppose that $H$ is a graph that has edge $(u, v)$. Since $u$ has $\Theta(d_1)$ neighbors in $Y$, there exist $\Theta(n - d_1)$ non-adjacent vertices of $Y$ to $u$. Let $C_Y$ denote the set of non-adjacent vertices of $Y$ to $u$. Each vertex in $C_Y$ has at least $\Theta(d_2)$ neighbors in $X$. Therefore, there $\Theta((n-d_1)d_2)$ candidate vertices for $x$. However, some of these edges from $y$ are already discovered by the algorithm. Note that the number of discovered edges is $o(n)$ at any point during the course of the algorithm because of the choice of $d$ in the construction. So by removing these $o(n)$ edges, there are still $\Theta((n-d_1)d_2)$ candidate for $x$. Furthermore, at most $\Theta(d_2^2)$ of the edges from a vertex of $C_Y$ to candidates for $x$, have an incident edge such that one of their endpoints of the incident edge is $v$. Therefore, there are at least $\Theta((n-d_1)d_2 - d_2^2)$ induced subgraphs of four vertices with the required properties.

Note that according to the construction, either $d_1 = \Theta(d)$ and $d_2 = \Theta(d)$, or $d_1 = \Theta(\log^4 n)$ and $d_2 = \Theta(\log^4 n)$ which implies that $\Theta((n-d_1)d_2 - d_2^2) = \Theta(nd_2)$. We couple subgraph $H$ to all $\Theta(nd_2)$ graphs that are obtained by removing edges $\{(u,v), (x,y)\}$ and adding edges $\{(u,y), (x,v)\}$. On the other hand, each graph of $\hat{\mathcal{G}}$ is coupled with $\Theta(d_1 d_2)$ graphs in $\mathcal{G}$ since degree of $u$ is $\Theta(d_1)$ and degree of $v$ is $\Theta(d_2)$. Therefore, we have $|\mathcal{G}|/|\hat{\mathcal{G}}| \leq O(d_1/n)$ which completes the proof.

\subsection{The Label Guessing Game on Trees}

\begin{claim}\label{clm:no-broken-vertex-discovered}
    Any algorithm $\mathcal{A}$ that makes at most $Q = o(d^{1/\epsilon})$ adjacency list queries, does not discover any broken vertex with high probability.
\end{claim}

\begin{proof}
By \Cref{obs:broken-vertex-upperbound}, there are at most $O(\sqrt{nd} \log n)$ broken vertices. Therefore, if we choose a random vertex, the probability of being a broken vertex is at most $O(\sqrt{d}\log n/\sqrt{n})$. Also, by the same 2-switch technique as the proof of \Cref{lem:edge-probability-upper-bound}, we can show that when we query a neighbor of a vertex, the probability of being broken is almost the same as when we choose a vertex uniformly at random. Since the algorithm makes at most $o(d^{1/\epsilon})$ queries, the total probability of finding a broken vertex is $O(d^{1/\epsilon}\sqrt{d}\log n/\sqrt{n}) = o(1)$.
\end{proof}

\begin{corollary}\label{cor:neighbors-prob}
    Let us condition on the high probability event of \Cref{clm:no-broken-vertex-discovered} that none of the broken vertices has been queried by the algorithm. Suppose that the algorithm makes a query to the adjacency list of vertex $v$ that is in subset $X$ and $u$ is the answer to the query. Then, the subset $Y$ that $u$ belongs to is determined by the binomial distribution that is defined in the construction.
\end{corollary}

\begin{proof}
    Fix a vertex $v$. Note that the type of neighbor of $v$ that is connected to $v$ by a non-broken edge is determined by a binomial random variable that is defined in the construction.
\end{proof}

By conditioning on the high probability event of \Cref{lem:rooted-forest} that the queried edges make a rooted forest and the properties of the input distribution, by \Cref{cor:neighbors-prob}, we can assume that we are in a tree model where each vertex has a label according to the subset that it belongs to and the distribution coming from the following transition probabilities. This is exactly the label guessing game outlined in the technical overview of \cref{sec:techniques-label-guessing} (see \cref{fig:tree}). 

\paragraph{Labels of vertices:} We use $S$ as the label of vertices in subset $S^1$ and $S^2$, $A_i$ for vertices in subset $A^1_i$ and $A^2_i$, $B_i$ for vertices in subset $B^1_i$ and $B^2_i$, and $D_i$ for vertices in subset $D_i$ for $i \in [1/\epsilon]$.

\paragraph{Transition probabilities:} Suppose that we condition on the high probability event that the algorithm does not query any broken vertex. Let $(u,v)$ be an edge in the forest that is queried by the algorithm and $u$ is the parent of $v$. Then, if $u$ has label $X$ and $v$ has label $Y$ for $X, Y \in \bigcup_{i=1}^{1/\epsilon}\{A_i, B_i, D_i\} \cup \{S\}$, then the transition probabilities from label $X$ to label $Y$ is according to the binomial distribution for neighbors of $X$ in \Cref{sec:input-distribution}.

\section{Indistinguishability of the Label of the Root}\label{sec:indis-label-root}
In this section, we show that if the result of the queried edges is a rooted tree of size $Q = o(d^{1/(2\epsilon)})$, then the algorithm can distinguish the label of the root with probability at most $\widetilde{O}(Q^2/d^{1/\epsilon})$ if the label of the root is in $\{A_{1/\epsilon}, B_{1/\epsilon}, D_{1/\epsilon}\}$. Our proof consists of two parts. First, we show that when we have $o(d^{1/(2\epsilon)})$ queries in the tree, with probability at least $1 - \widetilde{O}(Q^2/d^{1/\epsilon})$, all paths that start from the root and reach an $S$ vertex must contain a {\em mixer vertex} that we define later in the section. We define mixer vertices such that if a path contains such a vertex, then the algorithm does not learn anything about the label of the root from this path. For this, we prove a stronger claim that starting from the root of the tree, there is no path that contains more than $1/\epsilon - 1$ special edges before crossing a mixer vertex.

Second, conditioning on the above event, we prove that the algorithm will see the same tree if the root is in $\{A_{1/\epsilon}, B_{1/\epsilon}, D_{1/\epsilon}\}$, which implies that the algorithm cannot distinguish the label of the root with probability at least $1 - \widetilde{O}(Q^2/d^{1/\epsilon})$.

\begin{definition}[Special Edges]\label{def:special-edge}
    We call an edge $(u,v)$ special, if one of the following holds:
    \begin{itemize}
        \item $u \in B_i$ and $v \in A_{i-1}$, or $u \in A_{i-1}$ and $v \in B_i$ for $1 < i \leq 1/\epsilon$,
        \item $u \in S$ and $v \in B_1$, or $u \in B_1$ and $v \in S$,
        \item Let $p = (1-2\epsilon + 2i\epsilon^2 - 5\epsilon^2 / 2 + 3\epsilon^4)d + \log^4 N$. For each vertex in $D_i$, each of its neighbors to $D_i$ has a probability of $\log^4 N / p$ to be special (in other words, we can assume that there is $\log^4 N$ regular graph of special edges in each $D_i$),
        \item $(u,v)$ is among edges that only exists in exactly one of \yesdist{} or \nodist{}.
    \end{itemize}
\end{definition}

\begin{definition}[Mixer Vertices]\label{def:mixer-vertex}
Let $T$ be a rooted tree and $u$ be its root. Suppose that we are given that $u \in \{A_{1/\epsilon}, B_{1/\epsilon}, D_{1/\epsilon} \}$. Let $v$ be a vertex in $T$ and suppose that there are $k$ special edges on the path between $u$ and $v$. If $k < 1/\epsilon - 1$, we say $v$ is a mixer vertex if and only if $v \in \bigcup_{i=1}^{1/\epsilon-k-1} D_j$.
\end{definition}

The following observation is directly implied by the \Cref{def:special-edge}, \Cref{def:mixer-vertex}, and the construction of the input distribution.

\begin{observation}\label{obs:required-special-edges}
    Let $T$ be a rooted tree that is queried by the algorithm and $u$ be its root where $u \in \{A_{1/\epsilon}, B_{1/\epsilon}, D_{1/\epsilon} \}$. If there exists a path from $u$ to an $S$ vertex that does not contain a mixer vertex, then it contains at least $1/\epsilon - 1$ special edges.
\end{observation}

The intuition behind defining mixer vertex this way is that if the root of the tree is a level $1/\epsilon$ vertex and on a path that the algorithm queries, if there are $k$ special edges, then all vertices with a level of at least $1/\epsilon - k$, have the same probability of having neighbors among vertices of $\bigcup_{j=1}^{1/\epsilon-k-1} D_j$ which implies that if the path crosses one of those mixer vertices, then the algorithm cannot distinguish the label of the root using that path.

\begin{lemma}\label{lem:mixer-vertex-in-tree}    Let $\mathcal{A}$ be any algorithm that makes at most $o(d^{1/\epsilon})$ queries and $T$ be one of the rooted trees queried by the algorithm. Moreover, assume that the root of the tree is a vertex with level $1/\epsilon$. Then, with probability at least $1 - \widetilde{O}(|V(T)|/d^{1/\epsilon - 1})$, all paths between the root and a vertex that does not contain a mixer vertex, have at most $1/\epsilon - 2$ special edges on it.
\end{lemma}

\begin{proof}
    First, we prove that each path that $\mathcal{A}$ finds to a vertex $v$ that contains $1/\epsilon - 1$ special edges on it, has a probability of $\widetilde{O}(1/d^{1/\epsilon - 1})$ of not having any mixer vertex on it. For a mixer vertex $v$ such that $v \in D_j$, we use $j$ to show the {\em index} of the mixer vertex. Assume that we have an oracle that each time the algorithm finds a path with $1/\epsilon - 1$ special edges, it either returns the path that does not contain a mixer vertex or returns the mixer vertex on the path that has the lowest index among all mixers on the path. 

    Consider a path from the root to an $S$ vertex and a time $t$ that the algorithm has not queried the whole path yet. Suppose that the algorithm has found at most $1/\epsilon - 2$ special edges until time $t$. This implies that this path does not reach level $1$ or an $S$ vertex yet according to the construction and \Cref{obs:required-special-edges}. By the transition probability of the tree model, the probability of querying a $D_1$ vertex from a vertex of level 2 or larger is constant, however, the probability of querying a special edge is $\widetilde{O}(1/d)$ which implies that with probability $\widetilde{O}(1/d)$ the path crosses the $(1/\epsilon -1)$-th special edge before crossing a mixer vertex of level 1. Thus, among all paths that cross at most $1/\epsilon - 2$ special edges and are going to reach the next special edge, only $\widetilde{O}(1/d)$ fraction of them do not pass through a mixer vertex of level 1. Therefore, $\widetilde{O}(1/d)$ of all paths that have $1/\epsilon - 1$ special edges, do not contain a mixer vertex of level 1.

    Now consider all paths that do not contain a mixer vertex of level 1. With the same argument, for each of these paths, the probability of crossing $(1/\epsilon - 2)$-th special edge before crossing a mixer vertex of level 2 is $\widetilde{O}(1/d)$. Therefore, since the oracle only reveals the mixer vertex with the lowest index, then the probability of having a path with $1/\epsilon - 1$ special edges that do not contain a mixer vertex is $\widetilde{O}(1/d^{1/\epsilon - 1})$. Since there are at most $|V(T)|$ paths from the root, we obtain the claimed bound.
\end{proof}

\begin{corollary}\label{cor:mixer-vertex-in-tree}    Let $\mathcal{A}$ be any algorithm that makes at most $o(d^{1/\epsilon - 1})$ queries and $T$ be one of the rooted trees queried by the algorithm. Moreover, assume that the root of the tree is a vertex with level $1/\epsilon$. Then, with probability at least $1 - \widetilde{O}(|V(T)|/d^{1/\epsilon - 1})$, all paths between root and $S$ vertices in the tree contain a mixer vertex.
\end{corollary}

\begin{proof}
    Note that if there exists a path between the root and an $S$ vertex that does not contain a mixer vertex, it must contain at least $1/\epsilon - 1$ special edges. To see this, the only way that a vertex from level $i$ can reach level $i-1$ is to either cross a mixer vertex or a special edge. Combining with \Cref{lem:mixer-vertex-in-tree} we get the claimed bound.
\end{proof}

\begin{corollary}\label{cor:forest-mixer}
    Let $\mathcal{T}$ be a set of root trees such that the roots of all its trees belong to level $1/\epsilon$. Also, let $r = \sum_{T \in \mathcal{T}} |V(T)|$, and assume that we have $r = o(d^{1/\epsilon - 1})$. Then, with probability at least $1 - \widetilde{O}(r/d^{1/\epsilon - 1})$, all paths between the roots of trees and a vertex that in the same tree that does not contain a mixer vertex, have at most $1/\epsilon - 2$ special edges on it.
\end{corollary}
\begin{proof}
    Let $T_1, T_2, \ldots, T_k$ be all trees in $F$. By \Cref{lem:mixer-vertex-in-tree}, for each tree $T_i$, the probability of having such a path is at most $\widetilde{O}(|V(T_i)|/d^{1/\epsilon - 1})$. Hence, using union bound, the probability of having no path with more $1/\epsilon - 2$ special edges without any mixer vertex is at most $\widetilde{O}(r/d^{1/\epsilon - 1})$ which completes the proof.
\end{proof}

\begin{lemma}\label{lem:same-tree-different-labels}
    Let $T$ be a tree that is queried by an algorithm $\mathcal{A}$ on a graph that is drawn from input distribution, where the root belongs to level $1/\epsilon$. Also, suppose that on each path from the root of the tree to a vertex in the tree, if there are at least $1/\epsilon - 1$ special edges, then there exists at least one mixer vertex on the path. Then, the probability of seeing the same tree is equal for all possible roots in $\{A_{1/\epsilon}, B_{1/\epsilon}, D_{1/\epsilon} \}$ up to $(1 + o(d^{1/(2\epsilon) + 1} /n))^{|T|}$ multiplicative factor.
\end{lemma}

\begin{proof}
    The proof is involved and we begin by identifying some properties of input distribution that are useful in the proof. Let $G = (V, E)$ be the input graph that is drawn from the input distribution. Note that for all vertices
    in the graph except $S$ vertices, when $\mathcal{A}$ queries a new edge, the probability of the edge being special is the same.
    \begin{observation}\label{obs:special-edge-prob}
        Let $u$ be an arbitrary vertex in a graph that is drawn from input distribution. Also, let $(u, v)$ be a new queried edge by $\mathcal{A}$. Then, the probability of $(u,v)$ being a special edge is $\log^4 N / d'$.
    \end{observation}
    \begin{proof}
        The proof follows by the transition probability of the tree model and the way we defined special edges in \Cref{def:special-edge}.
    \end{proof}

     Let $L_i = \{A_i, B_i, D_i\}$ for $i \in [1/\epsilon]$. Also, let $E_S$ be the set of all special edges defined in \Cref{def:special-edge}. Let $G_i = G[\bigcup_{j=i}^{1/\epsilon} L_j]$. Let $u$ and $v$ be two different vertices in $G_i$. One important property of our input distribution is that if we query a neighbor of $u$ and $v$ and the queried edge is not a special edge, then the probability that the queried neighbor is a vertex in $G_i$ is equal for both $u$ and $v$.

    \begin{claim}\label{clm:non-special-edges-prob}
    Let $u, v \in V(G_i)$ for some $i \in [1/\epsilon]$. Also, let $(u, u')$ and $(v, v')$ be two edges that are queried by $\mathcal{A}$ and both are not special edges. Then, $\Pr[u' \in V(G_i)] = \Pr[v' \in V(G_i)]$.
    \end{claim}
    \begin{proof}
        By the construction of distribution, there is no edge between $\{u,v\}$ and $\bigcup_{j = 1}^{i-1} A_j \cup B_j$. Furthermore, if $\mathcal{A}$ queries an edge of a vertex in $V(G_i)$, with probability $\epsilon^4 d / d'$ the neighbor is in $D_j$ for $j < i$. Thus, the probability of the neighbor being in $\bigcup_{j = 1}^{i-1}D_j$ is $(i-1)\epsilon^4 d / d'$. Therefore, we have $\Pr[u' \in V(G_i)] = \Pr[v' \in V(G_i)]$.
    \end{proof}

    Now we are ready to complete the proof. Let $\ell_1, \ell_2 \in \{A_{1/\epsilon}, B_{1/\epsilon}, D_{1/\epsilon} \}$ be two different labels for the root of the tree $T$. The proof is based on a one-to-one coupling argument that for each tree that is queried by the algorithm if $\ell_1$ is the label of the tree, $\mathcal{A}$ will see the same tree with an equal probability if it starts from label $\ell_2$.

    For a vertex $u$ in tree $T$ such that there exists no mixer vertex on its path to the root, we define the notion of {\em progress of a vertex}, i.e. $p_u$, which shows the number of special edges on the path of root to $u$. Because of the assumption in the lemma statement, we have $0 \leq p_u < 1/\epsilon - 1$ for all $u \in T$.  

    \begin{observation}\label{obs:progress}
        Let $T$ be a tree that is queried by $\mathcal{A}$ and its root is in level $1/\epsilon$. Also, let $u$ be a vertex such that there is no mixer vertex on the path of $u$ to the root. Then, we have $u \in V(G_{1/\epsilon-p_u})$.
    \end{observation}
    \begin{proof}
        Note that according to the construction of the input distribution, if there is no mixer vertex on the path, each special edge can be used for going at most one level down in the input graph. Therefore, if there are $p_u$ special edges on the path, then $u \in V(G_{1/\epsilon-p_u})$.
    \end{proof}

     \begin{observation}\label{obs:mixer-jump-prob-coup}
        Let $u$ be a vertex in $T$ such that there is no mixer vertex on the path of $u$ to the root. Suppose that $\mathcal{A}$ queries the adjacency list of $u$ and let $\ell$ be the label of the neighbor. Then, for each $j \in [1/\epsilon - p_u - 1]$ it holds that $\Pr[\ell = D_j] = \epsilon^4 d/d'$.
    \end{observation}

    \begin{proof}
        By \Cref{obs:progress}, we have $u \in V(G_{1/\epsilon-p_u})$. According to the transition probabilities of the tree model, for a vertex in $G_{1/\epsilon-p_u}$ the probability of seeing a neighbor with label $D_j$ is $\epsilon^4 d / d'$ for $j \in [1/\epsilon - p_u - 1]$.
    \end{proof}

    Let $\mathcal{L}_1$ be a labeling for $T$ that $\mathcal{A}$ sees when it starts from a root with label $\ell_1$. Let $e = (u,v)$ be an edge in $T$ such that $u$ is the parent of $v$. There are three possible types for $e$ if there is no mixer vertex on a path between the root and $u$: 1) the edge is a special edge, 2) $v$ is a mixer vertex, 3) $e$ is an edge in $G_{1/\epsilon - p_u} \setminus E_S$. We give a labeling $\mathcal{L}_2$ for the same tree where the root has label $\ell_2$ and all $S$ vertices have label $S$ and the probability that $\mathcal{A}$ sees this labeling is equal to the probability of seeing $\mathcal{L}_1$. We maintain the invariant that the progress of each vertex is the same in both labeling $\mathcal{L}_1$ and $\mathcal{L}_2$.

    Now we start to process edges one by one according to the ordering that $\mathcal{A}$ makes queries. If the queried edge $e = (u, v)$ is of type (2), suppose that the label of $v$ is $D_j$ for some $j \in [1/\epsilon - p_u - 1]$. We assign the same label $D_j$ to $v$ in $\mathcal{L}_2$. By \Cref{obs:mixer-jump-prob-coup}, because of the invariant that $u$ has the same progress in both labeling, then the probability of seeing label $D_j$ is the same for both labelings. Moreover, for the subtree below $v$, we assume that all the labels are the same since the label of $v$ is the same at this point in both $\mathcal{L}_1$ and $\mathcal{L}_2$. Also, the invariant still holds. 
    
    If the queried edge $e = (u, v)$ is of type (1), i.e. is a special edge, we assume that in labeling $\mathcal{L}_2$, the edge is also a special edge and determine the label accordingly. Note that $v$ cannot have label $S$ in $\mathcal{L}_1$ since in this case there is no mixer vertex on the path to $v$ which is a contradiction \Cref{cor:mixer-vertex-in-tree}. By \Cref{obs:special-edge-prob}, the probability of querying a special edge is the same in both labelings. Furthermore, the invariant still holds since $p_v = p_u + 1$ in this case and the progress of vertex $u$ is the same in both labelings.

    Finally, if the queried edge $e = (u, v)$ is of type (3), we assume that $v$ has a label that is drawn from crossing an edge of $G_{1/\epsilon - p_u} \setminus E_S$. By \Cref{clm:non-special-edges-prob}, the probability of crossing the edge of type (3) is the same for both labeling.  Also, the invariant still holds since we did not add a special edge, which completes our coupling argument. It is also important to note that by the proof of \Cref{lem:rooted-forest}, the probability that the vertex with a new label is among non-singleton vertices is at most $o(d^{1/(2\epsilon) + 1} /n)$. Since the total number of steps is at most $|T|$, the probability that the new labeling is also a forest is almost equal within $(1 + o(d^{1/(2\epsilon) + 1} /n))^{|T|}$ multiplicative factor.
\end{proof}

\section{Indistinguishability of Crucial Edges}\label{sec:indis-yes-no}
Since the algorithm can make $o(d^{1/(2\epsilon)})$ queries, it might discover several edges that only appear in \yesdist{} because $\widetilde{\Theta}(1/d)$ fraction of total edges is specific to the \yesdist{}. However, we show that no algorithm can distinguish those edges with high probability. More specifically, we prove that the probability of seeing the same forest in \nodist{} is the same as \yesdist{} up to a $1 + o(1)$ multiplicative factor.

Let $E_{Y} = \{(u,v) \lvert (u,v) \in E, u \in A^1_{1/\epsilon}, v \in A^2_{1/\epsilon} \}$, and $E_Y^Q$ be the subset of $E_Y$ that the algorithm discovers. Moreover, let $V_Y^Q = \{ v \mid (u, v) \in E_Y\}$ where $u$ is the parent of $v$ in the queried rooted forest by algorithm $\mathcal{A}$. Also, assume that we remove all vertices of $V_Y^Q$ that have at least one ancestor in the forest which is in $V_Y^Q$.

\begin{claim}\label{clm:no-bad-event}
    Let $F$ be a forest that is queried by an algorithm $\mathcal{A}$ using at most $Q = o(d^{1/(2\epsilon)})$ queries. Then, with probability at least $1 - \widetilde{O}(Q^2/d^{1/\epsilon})$, all paths between the vertices of $V_Y^Q$ and a vertex in its subtree that have more than $1/\epsilon - 2$ special edges, contain a mixer vertex. 
\end{claim}

\begin{proof}
    Since the algorithm makes at most $o(d^{1/(2\epsilon)})$ queries, then we have $|V_Y^Q| \leq Q \leq o(d^{1/(2\epsilon)})$. Moreover, for each vertex in $V_Y^Q$, its subtree contains at most $o(d^{1/(2\epsilon)})$ vertices. Hence, the total number of vertices in all subtrees with root in $V_Y^Q$ is at most $Q^2 = o(d^{1/\epsilon})$. Therefore, by \Cref{cor:forest-mixer}, with probability at least $1 - \widetilde{O}(Q^2/d^{1/\epsilon})$, there exists a mixer vertex on all paths between vertices of $V_Y^Q$ and vertices of its subtree before crossing $1/\epsilon - 1$ special edges.
\end{proof}

We define a {\em bad event} to be the event that $\mathcal{A}$ finds a path between a vertex in $V_Y^Q$ and a vertex in its subtree that has more than $1/\epsilon - 2$ special edges without any mixer vertex. By \Cref{clm:no-bad-event}, since we assume that the algorithm makes $o(d^{1/(2\epsilon)})$ queries, the bad event happens with probability $o(1)$.

\begin{lemma}\label{lem:main-lemma-forest}
    Let us condition on having no bad event as we defined above. Let $\mathcal{A}$ be an algorithm that makes at most $Q = o(d^{1/(2\epsilon)})$ queries and $F$ be a rooted forest that is discovered by $\mathcal{A}$ on a graph that is drawn from \yesdist{}. Then, the probability of querying the same forest in the graph that is drawn from \nodist{} is equal to \yesdist{}.
\end{lemma}

\begin{proof}
    We say an edge is {\em crucial} in \yesdist{}, if the edge is in induced subgraph $G[A^1_{1/\epsilon}, A^2_{1/\epsilon}]$. In other words, in \yesdist{}, edges of $E_Y$ that we defined in this section are crucial edges. We extend the definition of crucial edges to have all edges that are specific to the \yesdist{}. Also, for \nodist{}, we define crucial edges to be the set of edges that only exists in \nodist{}. Note that ignoring the crucial edges, if we query an edge, the probability of the neighbor is the same in both distributions.

    Let the height of an edge in the forest be the distance of its closest endpoint to the root of the tree that belongs to. We prove this lemma using coupling between \yesdist{} and \nodist{}. We iterate over the height of the tree in decreasing order and inductively we show that we can switch from \yesdist{} to \nodist{}. Consider height $i$ in all trees. If the edge is not crucial, both distributions will sample similarly according to the construction. Since crucial edges are between vertices of level $1/\epsilon$, since we condition on not having a bad event for vertices of $V_Y^Q$, the subtree below the crucial edges are the same regardless of their labels up to a factor of $(1 + o( d^{1/(2\epsilon) + 1} /n))^{|T|}$ by \Cref{lem:same-tree-different-labels} if $|T|$ shows the size of the subtree. Since the total number of vertices in all these subtrees are at most $o(d^{1/\epsilon})$, and $d=n^{\epsilon/3}$,  the probability of discovering the
same forest in both distributions is equal up to a $1 + o(1)$ multiplicative factor. 
\end{proof}

Now we are ready to finish the proof of \Cref{thm:sublinear}.
\begin{proof}[Proof of \Cref{thm:sublinear}]
    By \Cref{clm:no-bad-event}, the probability of having a bad event is $o(1)$. If there is no bad event in the forest that the algorithm queries, the algorithm will discover the same forest with almost equal probability in both \yesdist{} and \nodist{}, by \Cref{lem:main-lemma-forest} with $1+o(1)$ multiplicative factor. Therefore, combining with \Cref{cor:matching-size-query}, any algorithm that computes a $(1, \epsilon n /7)$-approximate maximum matching, must make at least $d^{1/(2\epsilon)}$ queries. Also, by \Cref{rem:bipartite-note}, we can assume that each $D_i$ consists of two parts where one of them is connected to label $B$ vertices, the other one connected to label $A$, and there is no edge inside the induced subgraph of each of the two copies which implies that the graph is bipartite. Choosing $\epsilon' = \epsilon/7$ and combining it with the fact that $\Delta < 2d$ concludes the proof. 
\end{proof}

\begin{proof}[Proof of \Cref{thm:main}]
     Suppose for the sake of contradiction that there exists an LCA that computes $(1, \epsilon n)$-approximate maximum matching of $G$ with running time of $\Delta^{o(1/\epsilon)}$. We sample $t$ random vertices in the graph, and run this LCA on the selected vertices. Let $t'$ be the number of samples that the LCA returns a match for. We return $(t'/t) \cdot n/2$ as our estimate for the size of maximum matching. A simple Chernoff bound (see e.g. \cite{behnezhad2021,YoshidaYISTOC09}) shows that setting $t = \Theta(1/\epsilon^2)$ suffices for an estimation that is accurate up to an additive error of $\epsilon n$. From this, we get that there must exist an algorithm that runs in  $(1/\epsilon^2) \cdot \Delta^{o(1/\epsilon)}$ time and $(1, 2\epsilon n)$ approximates the size of maximum matching. Since we ruled out the existence of such an algorithm in \Cref{thm:sublinear}, there exists no such LCA.
\end{proof}

\paragraph{Acknowledgements.} Aviad Rubinstein was supported by David and Lucile Packard Fellowship.

\bibliographystyle{plainnat}
\bibliography{references}
	
\end{document}